\documentclass[11pt]{llncs}

\usepackage{graphicx}
\usepackage{amsmath}
\usepackage{amsfonts}
\usepackage{amssymb}

\newcommand{\be}{\begin{equation}}
\newcommand{\ee}{\end{equation}}
\newcommand{\ba}{\begin{eqnarray}}
\newcommand{\ea}{\end{eqnarray}}
\newcommand{\ban}{\begin{eqnarray*}}
\newcommand{\ean}{\end{eqnarray*}}

\newcommand*{\EC}{\mathrm{EC}}
\newcommand*{\leak}{\mathrm{leak}}
\newcommand*{\eps}{\varepsilon}
\newcommand*{\epsb}{\bar{\eps}}

\newcommand{\demi}{\frac{1}{2}}

\newcommand{\one}{\leavevmode\hbox{\small1\normalsize\kern-.33em1}}

\newcommand*{\id}{\mathrm{id}}

\newcommand*{\cE}{\mathcal{E}}
\newcommand*{\cF}{\mathcal{F}}

\newcommand*{\cX}{\mathcal{X}}

\newcommand*{\cZ}{\mathcal{Z}}

\newcommand*{\GF}{\mathrm{GF}}

\begin{document}

\title{Security Bounds for Quantum Cryptography with Finite Resources} \author{Valerio
  Scarani$^1$ and Renato Renner$^2$} \institute{$^1$ Centre for
  Quantum Technologies and Department of Physics, National University of Singapore, Singapore\\
  $^2$ Institute for Theoretical Physics, ETH Zurich, Switzerland}
\date{\today}

\maketitle

\begin{abstract}
  A practical quantum key distribution (QKD) protocol necessarily runs
  in finite time and, hence, only a finite amount of communication is
  exchanged.  This is in contrast to most of the standard results on
  the security of QKD, which only hold in the limit where the number
  of transmitted signals approaches infinity. Here, we analyze the
  security of QKD under the realistic assumption that the amount of
  communication is finite. At the level of the general formalism, we
  present new results that help simplifying the actual implementation
  of QKD protocols: in particular, we show that symmetrization steps,
  which are required by certain security proofs (e.g., proofs based on
  de Finetti's representation theorem), can be omitted in practical
  implementations. Also, we demonstrate how two-way reconciliation
  protocols can be taken into account in the security analysis. At the
  level of numerical estimates, we present the bounds with finite
  resources for ``device-independent security'' against collective
  attacks.
\end{abstract}

\section{Introduction}

Quantum key distribution (QKD) is one of the most mature fields of quantum information science, both from the theoretical and the experimental point of view \cite{review1,review2,review3}. This does not mean, however, that the open questions are merely technical ones: in this paper, we are concerned with an issue that is in fact rather crucial for the assessment of security of real devices.

Most unconditional security proofs of QKD have provided an
asymptotic bound for the secret key rate $r$, valid only in the limit
of \textit{infinitely long
  keys}~\cite{ShoPre00,Mayers01,LoChau,Koashi,BenOr02}. This reads in
general~\cite{DevWin05} \ba r & =& S(X|E) - H(X|Y) \
, \label{eq:entropyrateasym} \ea where $S(X|E) := S(X E) - S(E)$ and
$H(X|Y) := H(X Y) - H(Y)$ are the conditional von Neumann and Shannon
entropies, respectively, evaluated for the joint state of Alice and
Bob's raw key and the system controlled by Eve (after the sifting
step).

In real experiments, obviously, \textit{finite resources} are used. As
a matter of fact, the need for finite key analysis was recognized
several years ago~\cite{ina07}. In early security proofs though, the
\textit{security parameter} \ba \mbox{``Deviation from the ideal
  case''}&\leq& \eps\,.  \ea was defined in terms of ``accessible
information''. This measure of deviation had two shortcomings, namely
(i) it does not provide composable security, as proved
in~\cite{KRBM07}, and (ii) it has no operational interpretation. It
turns out that both shortcomings are not problematic for asymptotic
bounds\footnote{The absence of an operational interpretation of $\eps$
  is not a problem since any deviation is supposed to vanish for
  asymptotically long keys. Furthermore, the fact that asymptotic
  bounds can be ``redeemed'' for composability is a consequence of the
  result of~\cite{RenKoe05} saying that keys obtained by two-universal
  hashing provide composable security.}, but for finite-key analysis a
different definition must be used. A correct definition was used for
the first time in~\cite{mey06}, but the authors considered only a
restricted class of attacks. While partial, these and other studies \cite{LoChauArdehali,Ma05,Wang05} triggered the awareness that a large $N$ would be required for a QKD experiment to produce a secure key.

More recently, Hayashi used a valid definition (although the concern
for composable security is not addressed explicitly) in his analysis
of the BB84 protocol with decoy states~\cite{hay2}. Hayashi's bound
has been applied to experimental data \cite{hase}. Apart from being
possibly the first creation of a truly unconditional secure key, this
experiment provides an instructive example of how critical finite key analysis is. Indeed, for the observed error rate $Q\approx 5\%$ and
the choice $\eps=2^{-9}$, $4100$ secret bits could be extracted from
each raw key block of $n\approx \frac{N}{2}=10^5$ bit: in other words,
the final secret key rate was $r\approx 2\%$, instead of the $r\approx
43\%$ predicted by the asymptotic bound. Security bounds for finite
resources are definitely one of the most urgent tasks for practical
QKD \cite{review3}.

Recently we have shown that the theoretical tools developed by one of
us \cite{rennerthesis} can be used to provide a compact
approach to security proofs in the non-asymptotic limit
\cite{scaren}. Our formalism leads to a generalized version of the
secret key rate that reads 
\ba r &=&(n/N)\,\left[
  S_\xi(X|E)-\Delta-\leak_{\EC}/n\right]\,.
\label{raten}\ea Comparing with (\ref{eq:entropyrateasym}), four
modifications should be noticed: (i) only a fraction $n$ of the
signals contributes to the key, the rest must be used for parameter
estimation; (ii) the parameter estimation has finite precision $\xi$;
(iii) the task of privacy amplification itself has a security
parameter $\Delta$; and (iv) the error correction protocol may not
reach the Shannon limit, so $\leak_{\EC}\geq n H(X|Y)$.

In this paper, we revisit our previous work and improve it by two
important observations (Lemmas \ref{lem:symm} and \ref{lem:leakage}
below), then we present a new example of explicit calculation (Section
\ref{ssdevindep}).

\section{Basic definitions}

\subsection{Definition of security}

In the existing literature on QKD, not only the analysis, but also the
very \emph{definition} of security is mostly limited to the asymptotic
case; and we therefore need to revisit it here. Most generally, the
security of a key $K$ can be parametrized by its \emph{deviation}
$\eps$ from a \emph{perfect key}, which is defined as a uniformly
distributed bit string whose value is completely independent of the
adversary's knowledge.  In an \emph{asymptotic} scenario, a key $K$ of
length $\ell$ is commonly said to be \emph{secure} if this deviation
$\eps$ tends to zero as $\ell$ increases.  In the
\emph{non-asymptotic} scenario studied here, however, the deviation
$\eps$ is always finite.  This makes it necessary to attribute an
\emph{operational interpretation} to the parameter $\eps$. Only then
is it possible to choose a meaningful security threshold (i.e., an
upper bound for $\eps$) reflecting the level of security we are aiming
at. Another practically relevant requirement that we need to take into
account is \emph{composability} of the security definition.
Composability guarantees that a key generated by a QKD protocol can
safely be used for applications, e.g., as a one-time-pad for
message encryption. Although this requirement is obviously crucial for
practice, it is not met by most security definitions considered in the
literature~\cite{KRBM07}.

Our results are formulated in terms of a security definition that meets both requirements, i.e., it
is composable and, in addition, the parameter $\eps$ has an
operational interpretation. The definition we use was proposed in~\cite{BHLMO05,RenKoe05}:
for any $\eps \geq 0$, a key $K$ is said to be \emph{$\eps$-secure with respect to an adversary $E$} if the joint state $\rho_{K E}$ satisfies
\begin{align} \label{eq:sec}
    \frac{1}{2} \bigl\| \rho_{K E} - \tau_K \otimes \rho_E \bigr\|_1 \leq \eps \ ,
\end{align}
where $\tau_K$ is the completely mixed state on $K$. The parameter
$\eps$ can be seen as the maximum probability that $K$ differs from a
perfect key (i.e., a fully random bit string)~\cite{RenKoe05}.
Equivalently, $\eps$ can be interpreted as the \emph{maximum failure
  probability}, where failure means that ``something went wrong'',
e.g., that an adversary might have gained some information on $K$.
From this perspective, it is also easy to understand why the
definition is composable. In fact, the failure probability of any
cryptosystem that uses a perfect secret key only increases by (at
most) $\eps$ if we replace the perfect key by an $\eps$-secure key. In
particular, because one-time pad encryption with a perfect key has
failure probability $0$ (the ciphertext gives zero information about
the message), it follows that one-time-pad encryption based on an
$\eps$-secure key remains perfectly confidential, except with
probability at most~$\eps$.

\subsection{Description of the Generic Protocol}

Although most practical quantum key distribution protocols are
\emph{prepare-and-measure} schemes, for analyzing their security it is
often more convenient to consider an \emph{entanglement-based}
formulation. In fact, such a formulation can be obtained by simply
replacing all classical randomness by quantum entanglement and
postponing all measurements. In the following, we describe the general
type of protocol our analysis applies to.

\begin{enumerate}
\item \emph{Distribution of quantum information:} Alice and Bob
  communicate over an (insecure) quantum channel to generate $N$
  identical and independent pairs of entangled particles.\footnote{We
    use the term \emph{particle} here only for concreteness. More
    generally, they might be arbitrary subsystems.}  The joint state
  of the $N$ particle pairs together with the information that an
  adversary might have on them (e.g., acquired by eavesdropping) is
  denoted by $\rho_{A^N B^N E^N}$.
\item \emph{Parameter estimation:} Alice and Bob apply a
  LOCC-measurement\footnote{A \emph{LOCC-measurement} is a measurement
    on a bipartite system that can be performed by local measurements
    on the subsystems combined with classical communication.} to $m$
  particle pairs selected at random (using the authentic communication
  channel). We denote the resulting statistics by $\lambda_m$ and the
  joint state of the remaining (not measured) particles and Eve's
  system by $\rho_{A^{N-m} B^{N-m} E^N}$. If the statistics
  $\lambda_m$ fails to satisfy certain criteria, Alice and Bob abort
  the protocol.
\item \emph{Measurement and advantage distillation:} Alice and Bob
  apply block-wise measurements $\cE_{A^b B^b}$ on their remaining
  particles to get raw keys $X^n$ and $Y^n$, respectively. More
  precisely, $\cE_{A^b B^b}$ is an arbitrary LOCC-measurement applied
  sequentially to blocks $A^b$ of $b$ particles on Alice's side and
  the corresponding particles $B^b$ on Bob's side. In a protocol
  without advantage distillation, $\cE_{A^b B^b} = \cE_A \otimes
  \cE_B$ simply consists of local measurements on single particles,
  i.e., $b=1$. However, $\cE_{A^b B^b}$ might describe any operation
  that can be performed by Alice and Bob on a finite block of particle
  pairs. The resulting state is then given by $\rho_{X^n Y^n E^N} =
  (\cE_{A^b B^b}^{\otimes n} \otimes \id_{E^N})(\rho_{X^{b n} Y^{b n}
    E^N})$, where $n$ is the number of blocks, i.e., $n b \leq N-m$.
\item \emph{Error correction:} Alice and Bob exchange classical
  messages, summarized by $C$, which allow Bob to compute a guess
  $\hat{X}^{b n}$ for Alice's string $X^{b n}$.
\item \emph{Privacy amplification:} Alice and Bob generate the final
  key by applying an appropriately chosen hash function to $X^{b n}$
  and $\hat{X}^{b n}$, respectively. The requirement on the hash
  function is that it maps strings with sufficiently high min-entropy
  to uniform strings of a certain length $\ell$ (such functions are
  sometimes called \emph{strong (quantum) extractors}). A typical (and
  currently the only known) class of functions satisfying this
  requirement are \emph{two-universal hash functions} (see
  Section~\ref{sec:hashing} for examples of two-universal function
  families).
\end{enumerate}

\section{Security analysis}

\subsection{Security against collective attacks}
\label{sscollective}

\newcommand*{\epsE}{\bar{\eps}}

An attack is said to be \emph{collective} if the interaction of Eve
with the quantum channel during the distribution step is i.i.d. This
implies that the state after the distribution step is i.i.d., too,
that is, $\rho_{A^N B^N E^N} = \sigma_{A B E}^{\otimes N}$, where
$\sigma_{A B E}$ is the density operator describing a single particle
pair together with the corresponding ancilla $E$ held by Eve.

The following analysis is subdivided into four parts. Each part gives
rise to separate errors, denoted by $\eps_{\mathrm{PE}}$, $\epsE$,
$\eps_{\mathrm{EC}}$, and $\eps_{\mathrm{PA}}$, respectively. These
sum up to
\begin{align} \label{eq:epssum}
  \eps = \eps_{\mathrm{PE}} + \epsE + \eps_{\mathrm{EC}} + \eps_{\mathrm{PA}} \ ,
\end{align}
where $\eps$ is the security of the final key (cf.~\eqref{eq:sec} for
the definition of security). Making the individual contributions
smaller comes at the cost of reducing other parameters that,
eventually, result in a reduction of the size of the final key (see
equations~\eqref{eq:Gamma}, \eqref{eq:iid}, \eqref{eq:EC}, and
\eqref{eq:PA}).

\begin{itemize}
\item{\textit{Parameter estimation (minimize set of compatible states
      $\Gamma$ and number of sample points~$m$ vs.\ minimize failure
      probability~$\eps_{\mathrm{PE}}$).}}

  Parameter estimation allows Alice and Bob to determine properties of
  $\sigma_{A B}$. We express this by defining a set
  $\Gamma_{\eps_{\mathrm{PE}}}$ containing all states $\sigma_{A B}$
  that are \emph{compatible} with the outcomes of the parameter
  estimation. For concreteness, we assume here that Alice and
  Bob|depending on the statistics of their measurements|either
  continue with the execution of the protocol or abort. The set
  $\Gamma_{\eps_{\mathrm{PE}}}$ is then defined as the set of states
  $\sigma_{A B}$ for which the protocol continues with probability at
  least $\eps_{\mathrm{PE}}$ (i.e., the states from which a key will
  be extracted with non-negligible probability). The quantity $\eps_{\mathrm{PE}}$ corresponds therefore to the probability that the parameter estimation passes although the raw key does not contain sufficient secret correlation.  
  In particular, if
  Alice and Bob continue the protocol whenever they observe a
  statistics $\lambda_m$ using a POVM with $d$ possible outcomes then
  (Lemma 3 of \cite{scaren})
\begin{align} \label{eq:Gamma}
  \Gamma_{\eps_{\mathrm{PE}}}
\subseteq
  \left\{\sigma_{A B}:\,\|\lambda_m-\lambda_{\infty}(\sigma_{A B})\|\leq {\textstyle \sqrt{ \frac{2 \ln(1/\eps_{\mathrm{PE}}) + d \, \ln(m+1)}{m}}}\right\} \ 
\end{align}
where $\lambda_{\infty}(\sigma_{A B})$ denotes the (perfect)
statistics in the limit of infinitely many measurements.

\item{\textit{Calculation of the min-entropy (minimize decrease of
  min-entropy~$\delta$ vs.\ minimize error probability~ $\epsE$).}}

Under the assumption of collective attacks, the joint state of Alice
and Bob's as well as the relevant part of Eve's system after the
measurement and advantage distillation step is of the form
$\rho_{X^{n} Y^{n} E^{b n}} = \sigma_{X Y E^b}^{\otimes n}$ where 
\begin{align} \label{eq:AD}
\sigma_{X Y E^b} 
:= 
({\cE_{A^b B^b} \otimes \id_{E^b}})(\sigma_{A B E}^{\otimes b})
\end{align}
This property allows to compute a lower bound on the smooth
min-entropy of $X^n$ given Eve's overall information $E^N$ (before
error correction), which will play a crucial role in the analysis of
the remaining part of the protocol. More precisely, the min-entropy
can be expressed in terms of the von Neumann entropy $S$ evaluated for
the state $\sigma_{X E^b}$,
\begin{align} \label{eq:iid}
  H_{\infty}^{\epsE}(X^n | E^N) \geq n (S(X|E^b)_{\sigma_{X E^b}} - \delta)
\end{align}
where $\delta := 7
\sqrt{\frac{\log_2(2/\epsE)}{n}}$.

\item{\textit{Error correction (information leakage~$\mathrm{leak}$ vs.\
  failure probability~ $\eps_{\mathrm{EC}}$).}}

Error correction necessarily involves communication $C$ between Alice
and Bob. The maximum leakage of information to an adversary is
expressed in terms of min- and max-entropies,
\begin{align*}
  \mathrm{leak} := H_0(C) - H_{\infty}(C|X^{n} Y^{n}) \ .
\end{align*}
While $H_0(C)$ corresponds to the total number of relevant bits
exchanged during error correction, we subtract $H_{\infty}(C|X^n Y^n)$
which is the number of bits that are \emph{independent} of the raw key
pair $(X^n, Y^n)$.  Note the formal resemblance of this expression to
the mutual information $I(C: X^n Y^n)$. Indeed, the quantity
$\mathrm{leak}$ counts the number of bits of $C$ that are
\emph{correlated} to the raw key. In particular, any information that
is independent of the raw key, such as the description of an error
correcting code, does not contribute. Also, in a protocol where
redundant messages are exchanged (this is for instance the case for
two-way error correction schemes such as the Cascade
protocol~\cite{cascade}), the quantity $\mathrm{leak}$ is generally
much smaller than the total number of communicated bits.

Typically, there is a trade-off between the leakage $\mathrm{leak}$
and the failure probability, i.e., the maximum probability that
$\hat{X} \neq X$ (where the maximum is taken over all possible states
in $\Gamma_{\eps_{\mathrm{PE}}}$), which we denote by
$\eps_{\mathrm{EC}}$. This trade-off depends strongly on the actual
error correction scheme that is employed, but typically has the form
\begin{align} \label{eq:leakec}
  \mathrm{leak}_{\eps_{\mathrm{EC}}} = f H_0(X|Y) + \log_2\frac{2}{\eps_{\mathrm{EC}}} 
\end{align}
where $f$ is a constant larger than $1$. In theory, there are error
correction schemes with $f$ arbitrarily close to $1$, but the decoding
is usually not feasible due to computational limitations. In practice,
$f\approx 1.05-1.2$.

\item{\textit{Privacy amplification (maximize final key length~$\ell$
  vs.\ minimize failure probability~$\eps_{\mathrm{PA}}$).}}

To evaluate the final key size, we need to bound the decrease of min-entropy after the leakage of information that occurred in error correction. It follows from Lemma~\ref{lem:leakage} below
that the smooth min-entropy of $X^n$ given Eve's
information after error correction is bounded by
\begin{align} \label{eq:EC}
  H_{\infty}^{\epsE}(X^n | E^N C) \geq H_{\infty}^{\epsE}(X^n | E^N)
  - \mathrm{leak}_{\eps_{\mathrm{EC}}} \ .
\end{align}
The security of the final key only depends on this quantity and the
efficiency of the hash function used for privacy amplification. More
precisely, if two-universal hashing\footnote{Two-universal hashing is
  the procedure normally used for privacy amplification.} is used
then, for any fixed $\eps_{\mathrm{PA}} > 0$, the maximum length
$\ell$ of the final key is bounded by
\begin{align} \label{eq:PA}
  \ell \leq H^{\epsE}(X^n | E^N C) - 2 \log_2 \frac{1}{\eps_{\mathrm{PA}}} \
  .
\end{align}

\end{itemize}

Combining~\eqref{eq:iid}, \eqref{eq:EC} and~\eqref{eq:PA}, we
conclude that the final key is $\eps$-secure, for $\eps =
\eps_{\mathrm{PE}} + \epsE + \eps_{\mathrm{EC}} + \eps_{\mathrm{PA}}$
as in~\eqref{eq:epssum}, if
\begin{align} \label{eq:ellcrit} \ell \leq n \left[\min_{\sigma_{A B
      E} \in \Gamma_{\eps_{\mathrm{PE}}}} S(X|E^b)_{\sigma_{X E ^b}} - \delta(\epsb)\right] - \mathrm{leak}_{\eps_{\mathrm{EC}}} - 2
  \log_2\frac{1}{\eps_{\mathrm{PA}}}
\end{align}
where $\sigma_{X E^b}$ is related to $\sigma_{A B}$ via~\eqref{eq:AD}
applied to a purification of $\sigma_{A B}$ and where $\delta(\epsE)=
7 \sqrt{\frac{\log_2(2/\epsE)}{n}}$.

\subsection{Security analysis against general attacks}

A general method to turn a proof against collective attacks into a
proof against the most general coherent attacks is to introduce
additional symmetries. Here we highlight two aspects
that have been dealt with only partially in previous works.

\paragraph{A Lemma on symmetrization.}

The following lemma states that the smooth min-entropy of the state
before the symmetry operations have been applied is lower bounded by
the smooth min-entropy of the symmetrized state.

\begin{lemma}\label{lem:symm}
  Let $\rho_{X E}$ be a cq-state and let $\{f_R\}$ be a family of
  functions on $X$. Then, for any $\eps \geq 0$ and $R$ chosen at
  random
  \begin{align*}
    H_{\infty}^{\eps}(X|E) \geq H_{\infty}^{\eps}(f_R(X) | E R) \ .
  \end{align*}
\end{lemma}

\begin{proof}
  The statement is proved by sequentially applying rules of the smooth
  entropy calculus.
  \begin{align*}
    H_{\infty}^\eps(X|E)
    & = H_{\infty}^\eps(X|E) + H_{\infty}(R | R) \\
    & = H_{\infty}^\eps(X R | E R) \\
    & = H_{\infty}^\eps(f_R(X) X R | E R) \\
    & \geq H_{\infty}^\eps(f_R(X) | E R) \ .
  \end{align*}
  The first equality holds because $H_{\infty}(R|R) = 0$ (there is no
  certainty about $R$ if $R$ is known), and the second is a
  consequence of the additivity of the min-entropy (Lemma~3.1.6
  of~\cite{rennerthesis}). The third equality is a simply consequence
  of the fact that the computation of the value $f_R(X)$ while keeping
  the input is a unitary operation, under which the min-entropy is
  invariant. Finally, the inequality holds because tracing out the
  classical systems $X$ and $R$ can only decrease the smooth
  min-entropy (see Lemma~3.1.9 of~\cite{rennerthesis}).
\end{proof}

An important practical consequence of this Lemma is that \textit{the symmetrization needs not be actually implemented}. Indeed, the smooth min-entropy is basically the only quantity that is relevant for the
security of the final key: then, the statement of the Lemma implies that, if the
symmetrized version of the protocol is secure, the original version is also secure.

\paragraph{Permutation symmetry.}

Lemma \ref{lem:symm} above is valid for any symmetrization. Typically, one considers permutation symmetry. This can be achieved, for instance, by randomly permuting the positions of the bits~\cite{rennerthesis} (more precisely, Alice and Bob both apply the same, randomly chosen, reordering to their bitstring). The symmetric states
can then be shown to have properties similar to those of i.i.d.~states, e.g. via the quantum de Finetti theorem~\cite{symindep}. This in turn leads to a bound of the form \eqref{eq:iid}, with a different definition of the parameter $\delta$ (cf. Theorem 6.5.1 in \cite{rennerthesis}, referring to Table 6.2 for the parameters; the corrections due to the de Finetti theorem are the terms that involve the quantities $k$ and $r$). Thus, a lower bound for security using finite resources can be computed for any discrete-variable protocol.

Such a bound turns out to be very pessimistic: this is the price to
pay for its generality\footnote{Also, it is an open question whether
  the existing de Finetti theorem provides tight estimates, or if the
  bounds can be improved.}. When considering some specific protocols,
there can be other, more efficient ways to obtain i.i.d. Specifically,
for the BB84 \cite{bb84} and the six-state protocol
\cite{benbra2,bru98,bec99}, suitable symmetries can be implemented in the
protocol itself by random but coordinated bit- and phase
flips~\cite{gotlo,KGR}. Security bounds against general attacks can be
computed by considering i.i.d.~states just because of these
symmetries, thus by-passing the need for the de Finetti theorem.

\subsection{Decrease of the smooth min-entropy by information leakage}
\label{ssleak}

An essential part of the technical security proof presented above is the following
lemma, which provides a bound on the decrease of the min-entropy by
information leakage in the error correction step. The statement shown
here is a generalization of a corresponding statement
in~\cite{rennerthesis}, which has been restricted to one-way error
correction.

\begin{lemma} \label{lem:leakage} The decrease of the smooth
  min-entropy by the leakage of information in the error correction
  step is given by
  \begin{align*}
    H_{\infty}^\eps(X | E C) \geq H_{\infty}^\eps(X | E) - \mathrm{leak} \ .
  \end{align*}
\end{lemma}

\begin{proof}
  \begin{align*}
    H_{\infty}^{\eps}(X| E C)
  & \geq
    H_{\infty}^{\eps}(X C | E) - H_0(C) \\
  & \geq
    H_{\infty}^{\eps}(X | E) + H_{\infty}(C | X E) - H_0(C) \\
  & \geq 
    H_{\infty}^{\eps}(X | E) + H_{\infty}(C | X Y E) - H_0(C) \\
  & =
    H_{\infty}^{\eps}(X | E) + H_{\infty}(C | X Y) - H_0(C) 
  \end{align*}
  The first two inequalities are chain rules and the third is the
  strong subadditivity for the smooth min-entropy. The last equality
  follows from the fact that $E \leftrightarrow (X,Y) \leftrightarrow
  C$ is a Markov chain, because the communication $C$ is computed by
  Alice and Bob.
\end{proof}

\subsection{Two-universal hashing} \label{sec:hashing}

As explained above, privacy amplification is usually done by
two-universal hashing. 

\begin{definition} \label{def:gentu} A set $\cF$ of functions $f$ from
  $\cX$ to $\cZ$ is called \emph{two-universal} if
  \begin{align*}
    \Pr_{f \in \cF}\bigl[f(x) = f(x')\bigr] \leq \frac{1}{|\cZ|} \ ,
  \end{align*}
  for any distinct $x, x' \in \cX$ and $f$ chosen at random from $\cF$
  according to the uniform distribution.
\end{definition}

To perform the privacy amplification step, the two parties simply have
to choose at random a function $f$ from a two-universal set $\cF$ of
functions that output strings of length $\ell$, where $\ell$ is chosen
such that it satisfies~\eqref{eq:ellcrit}. As shown below, there exist
constructions of two-universal sets $\cF$ of functions that are both
easy to describe (the description length is equal to the input length)
and that can be efficiently evaluated.

Examples of two-universal function families have first been proposed
by Carter and Wegman~\cite{CarWeg79,WegCar81}. One of the
constructions mapping $n$-bit strings to $\ell$-bit strings, for any
$\ell \leq n$, only involves addition and multiplication in the field
$\GF(2^n)$. It is defined as the family $\cF = \{f_r\}_{r \in
  \GF(2^n)}$ of functions $f_r$ that, on input $x$, output the $\ell$
least significant bits of $r \cdot x$ (where $\cdot$ denotes the
multiplication in $\GF(2^n)$), i.e.,
\begin{align*}
  \begin{array}{lcccl}
  f_{r} : \quad & \GF(2^n) & \longrightarrow & \GF(2^\ell)  \\
  & x & \longmapsto & [ r \cdot x ]_{\ell} &.
  \end{array}
\end{align*}

\section{Computing security bounds}

\subsection{Summary of the previous section}

Let us re-phrase the results obtained above in a more operational way. An experiment is characterized by the following parameters:
\begin{itemize}
\item The protocol, in particular $d$ the number of outcomes of the measurements;
\item The number of exchanged quantum signals $N$;
\item The estimates of the channel parameters; 
\item The performances of the error correction protocol, in particular $\eps_{\mathrm{EC}}$ and $f$ (recall that these are functions of the parameters);
\item The desired level of security $\varepsilon$.
\end{itemize}
We have found above the bound \eqref{eq:ellcrit} for the extractable secret key length $\ell$, which is valid for collective attacks, and also for general attacks in the case of the BB84 and the six-state protocols. By setting $r=\frac{\ell}{N}$, one gets the announced expression \eqref{raten} for the secret key rate.

The expression for $r$ is thus a function of the parameters listed above and several others, namely:
\begin{itemize}
\item $n$, $b$ and $m$, subject to the constraint $nb+m\leq N$;
\item $\eps_{\mathrm{PE}}$, $\epsb$ and $\eps_{\mathrm{PA}}$, subject to the constraint $\eps =
\eps_{\mathrm{PE}} + \epsE + \eps_{\mathrm{EC}} + \eps_{\mathrm{PA}}$.
\end{itemize}
The best value for $r$ is therefore obtained by optimizing \eqref{eq:ellcrit} over the free parameters\footnote{Note that a parameter may be free \textit{a priori} but be fixed in a given experiment. For instance, if in BB84 the choice of the basis is made passively through a 50-50 beam splitter, one has the additional constraint $m=nb$.}, for a given experiment.

In Ref.~\cite{scaren}, we have presented such an optimization for the BB84 and the six-state protocols implemented with single photons, under the restriction that $f$ is a constant and $b=1$ (one-way error correction). Here, we present the computation of the security bound with finite resources for another protocol.

\subsection{An application: ``device-independent security'' against collective attacks}
\label{ssdevindep}

In 1991, Ekert noticed that the security of QKD could be related to
the violation of Bell's inequalities \cite{eke91}. This remark
provided him with the basic intuition, but it remained purely
qualitative. Only recently, on a modified version of the Ekert
protocol \cite{aci06}, it has been possible to provide a quantitative
bound on Eve's information that depends only on the violation of a
particular Bell-type inequality \cite{aci07}. The remarkable property
of this study is that this bound is ``device-independent'': the
knowledge of (i) the dimension of the Hilbert space in which Alice's
and Bob's signals are encoded and of (ii) the details of the
measurements that are performed, is \textit{not} required. The price
to pay for such generality is that there is, as of today, no argument
to conclude to unconditional security\footnote{This is in particular true because one does not bound the dimension of the Hilbert space; so the available de Finetti theorem cannot be used. It is important to stress that the usual unconditional security bounds \textit{do} rely on the assumption that the dimension of the Hilbert space is known --- and this is actually more serious than just a technical assumption for the proofs: most protocols, like BB84 and six-state, become provably \textit{insecure} if one cannot rely on the fact that a meaningful fraction of the measurements are done on two-qubit signals.}: the bound has been proved only
for collective attacks. It is also worth stressing that, as long as the detection loophole remains open, device-independent security cannot be assessed on real setups \cite{aci07,zhao07}.

Using our approach, we are going to obtain the non-asymptotic bound for device-independent security against collective attacks. We can use \eqref{eq:ellcrit} directly. Two elements depend on the protocol and must be discussed:
\begin{itemize}
\item The relation between $n$ and $m$ depends on the measurements
  specified by the protocol (here we set $b=1$). The protocol
  specifies that Alice performs three measurements $A_0$, $A_1$ and
  $A_2$, while Bob performs two measurements $B_1$ and $B_2$. The key
  is extracted out of the events $(A_0,B_1)$. Coherence in the channel
  is checked by the Clauser-Horne-Shimony-Holt (CHSH) inequality
  \cite{chsh} using $(A_1,A_2;B_1,B_2)$, i.e. from the quantity \ba
{\cal C}&=&E(A_1B_1)+E(A_1B_2)+E(A_2B_1)-E(A_2B_2)\label{eq:chsh}
\ea where $E(A_iB_j)=\mbox{Prob}(a_i=b_j)-\mbox{Prob}(a_i\neq b_j)$ is the correlation coefficient for bits.
We suppose that Alice chooses
  $A_0$ with probability $p_{a0}$ and the other settings with equal
  probability $p_{a1}=p_{a2}=(1-p_{a0})/2$; and that Bob chooses $B_1$
  with probability $p_{b1}$ and $B_2$ with probability
  $1-p_{b1}$. Therefore \ba n=p_{a0}p_{b1}N&,&m_{ij}=\demi(1-p_{a0})p_{bj}N \ea and
  the other events are discarded.

\item In \eqref{eq:ellcrit}, only $S_\xi(X|E)\equiv\max_{\sigma_{A B E} \in \Gamma_{\eps_{\mathrm{PE}}}} S(X|E^b)_{\sigma_{X E ^b}}$ depends on the protocol, and this quantity contains only the imprecision of the parameter estimation as a finite-key effect --- indeed, the other three modifications due to the finite resources, listed in Section \ref{sscollective}, give rise to the other terms in \eqref{eq:ellcrit} that are independent of the protocol. Therefore, we only have to allow a deviation of the measured parameters by the quantity $\xi(m,d)=\sqrt{ \frac{2 \ln(1/\eps_{\mathrm{PE}}) + d \, \ln(m+1)}{m}}$ as defined in \eqref{eq:Gamma}. The asymptotic version \cite{aci07} \ba
S_{\xi=0}(X|E)&=&1-h\left(\frac{1+\sqrt{({\cal C}/2)^2-1}}{2}\right)\ea
depends only on ${\cal C}$ given in \eqref{eq:chsh}. Now, the deviation on the estimate of $E(A_iB_j)$ is $\xi(m_{ij},2)$ because a correlation coefficient can be measured by a POVM with $d=2$ outcomes (``equal bits'' and ``different bits''). The most unfavorable case being obviously the one when the true value of ${\cal C}$ is lower than the estimated one, we obtain
\ba S_{\xi}(X|E)&=&1-h\left(\frac{1+\sqrt{[({\cal C}-\xi)/2]^2-1}}{2}\right)\ea  with $\xi=\sum_{i,j=1}^2 \xi(m_{ij},2)$.
\end{itemize}
Having described the quantities that depend on the protocol, we can run the optimization of $r$ for any $N$ and for some chosen values of $\eps$, $\eps_{\EC}$, $f$ and the observed parameters (${\cal C}$ and the error rate $Q$). The result is plotted in Fig.~\ref{figdevindep}. Similarly to what observed for BB84 and six-states \cite{scaren}, no key can be extracted for $N\lesssim 10^5$, and the asymptotic value is reached only for $N\gtrsim 10^{15}$. By monitoring the parameters of the optimization, one finds also that $p_{a0}$ and $p_{b_1}$ tend to 1 in the limit $N\rightarrow\infty$, as expected.

\begin{figure}
\includegraphics[scale=0.6]{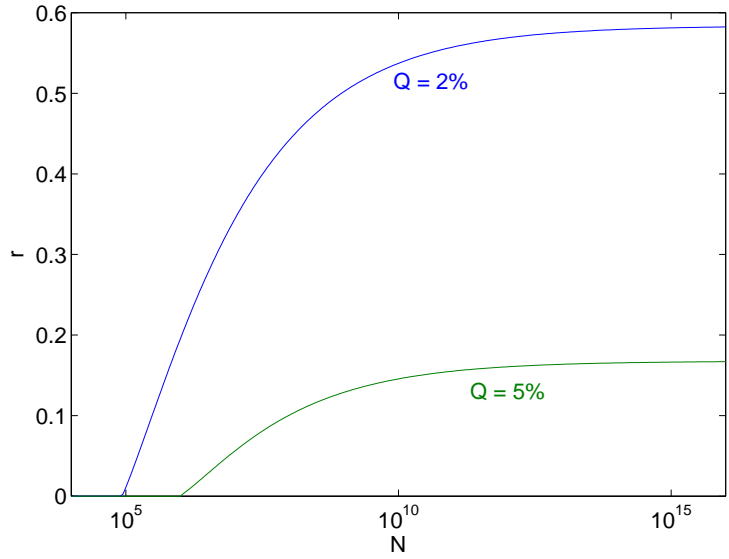}
\caption{Finite-key bound for device-independent security against collective attacks: secret key rate $r$ as a function of the number of exchanged quantum signals $N$, for two values of the observed error rate $Q$; we have assumed the relation ${\cal C}=2\sqrt{2}(1-2Q)$, which implies ${\cal C}\approx 2.715$ for $Q=2\%$ and ${\cal C}\approx 2.546$ for $Q=5\%$. We have fixed $\eps=10^{-5}$, $\eps_{\EC}=10^{-10}$ and $f=1.2$; we have supposed symmetric errors $\mbox{Prob}(a_0\neq b_1)=Q$, so that $H_0(X|Y)$ in \eqref{eq:leakec} is replaced by $h(Q)$.} \label{figdevindep}
\end{figure}

\section{Conclusion}

In this paper, we have built on our previous work on finite-key analysis \cite{scaren} and completed it with some important remarks. Lemma \ref{lem:symm} shows that the symmetrization of the data, although required to achieve security proofs, does not need to be done actively, because the min-entropy of the symmetrized data provides a bound for the min-entropy of the non-symmetrized ones. Lemma \ref{lem:leakage} extends our formalism to include two-way information reconciliation. After completing the general formalism with these Lemmas, we have applied it to derive a finite-key bound for device-independent security against collective attacks (Section \ref{ssdevindep}).

\textit{Acknowledgments.}--- This work is supported by the National
Research Foundation and Ministry of Education, Singapore.

\end{document}